\documentclass{article}
\usepackage{spconf}
\usepackage{blindtext}
\usepackage{graphicx}
\usepackage{amsfonts}
\usepackage{color}
\usepackage{amsmath}
\usepackage{amsthm} % for proof environment on ILIA
\usepackage{stmaryrd}
\usepackage{algorithm}
\usepackage{algpseudocode}
\usepackage{accents}

\usepackage{caption}
\newtheorem{conj}{Conjecture}
\newtheorem{prop}{Proposition}
\newtheorem{definition}{Definition}

\newcommand\tensor[1]{\underline{\mathbf{#1}}}
\DeclareMathAlphabet\tn{OMS}{cmsy}{b}{n}

\title{AUTHOR GUIDELINES FOR ICASSP 2018 PROCEEDINGS MANUSCRIPTS}
\name{Giuseppe G. Calvi, Ilia Kisil,  Danilo P. Mandic}
\address{Electrical and Electronic Engineering Department, Imperial College London, SW7 2AZ, UK\\ E-mails: \{giuseppe.calvi15, i.kisil15, d.mandic\}@imperial.ac.uk}
\begin{document}
	%\ninept
	\onecolumn
	\title{The Sum of Tensor networks}


	\markboth{Journal of \LaTeX\ Class Files,~Vol.~6, No.~1, January~2007}%
	{Shell \MakeLowercase{\textit{et al.}}: Bare Demo of IEEEtran.cls for Journals}



	\maketitle

	\begin{abstract}
		%\boldmath
		Tensor networks (TNs) have been gaining interest as multiway data analysis tools owing to their ability to tackle the curse of dimensionality and to represent tensors as smaller-scale interconnections of their intrinsic features. However, despite the obvious advantages, the current treatment of TNs as stand-alone entities does not take full benefit of their underlying structure and the associated feature localization. To this end, embarking upon the analogy with a feature fusion, we propose a rigorous framework for the combination of TNs, focusing on their summation as the natural way for their combination. This allows for feature combination for any number of tensors, as long as their TN representation topologies are isomorphic. The benefits of the proposed framework are demonstrated on the classification of several groups of partially related images, where it outperforms standard machine learning algorithms.
	\end{abstract}

	\begin{keywords}
		Sum of tensor networks, Tucker decomposition, classification, feature extraction, graphs
	\end{keywords}

	% \fi

	\section{Introduction}

	Tensors are multidimensional generalizations of matrices and vectors, and their ability to make enhanced use of data structures to perform dimensionality reduction and component extraction offers a powerful tool in the analysis of Big Data. Owing to their flexibility and a scalable way to deal with multi-way data, tensors have found application in a wide range of disciplines, from the most theoretical ones, such as mathematics and numerical analysis \cite{comon,lathauwer_hooi,lathauwer_hosvd}, to the more practical signal processing applications \cite{cichocki_brain,tenssignalprocessing}.

	Applying standard numerical methods to tensors may be difficult, as in the raw format the required storage memory and number of operations grow exponentially with the tensor order (\textit{curse of dimensionality}) \cite{ttoseledets}. To overcome this issue, tensor decompositions (TDs) have been introduced with the aim to represent tensors by a much smaller number of parameters, via multilinear operations over the latent factors. The most well-known TD approaches are the Canonical Polyadic \cite{parafac, parafac2}, the Tucker \cite{tucker1, tucker2}, and the Tensor Train decompositions \cite{ttoseledets} (CPD, TKD, and TT respectively). 

	Any TD can be considered as a special case of the more general concept of tensor networks (TNs), which represent a high order tensor as a set of sparsely interconnected core tensors and factor matrices \cite{tenssignalprocessing}. In other words, TNs can be viewed as multi-core interconnections of features of the original tensor. The advantages of representing a tensor as a TN are: (i) TNs are perfectly suited to deal with the curse of dimensionality, as a high order tensor can be stored on different machines which deal with only the individual cores, (ii) each core may be representative of specific characteristics of the underlying tensor, thus implying inherent feature extraction on the original data. Despite these advantages, open problems in practical design of TNs include: (i) the choice of the TD for a particular application, (ii) minimization of the number of the parameters necessary for the TN representation \cite{ttoseledets}, and (iii) to the best of our knowledge, a rigorous framework to combine TNs. 
	
	In this work, we address the point (iii), by focusing on the issue of TN summation for TNs of the same topology. Summation is the most natural way to combine any two entities into a new one, and we postulate that a sum of multiple tensors (summands), \textit{in their TN format}, preserves the underlying structure in the inherent features of the summands. In this way, the sum of TNs immediately yields another isomorphic TN, the cores of which are a combination of the corresponding cores of the summand TNs. The summation operator can hence be interpreted as a process of mixing features of the original tensors, however, algorithms for tensor network summation are still in their infancy. 
	
	To this end, we propose a novel framework for the summation of TNs (and inherently any two or more tensors represented by TNs). This is achieved by simple block arrangements of the corresponding original cores. Leveraging on the very efficient way in which TNs represent large tensors in terms of the required storage and computation, and realizing that interconnections among the cores in a TN describe how data structures of the original tensors are intertwined, we explore the possibility to combine \textit{corresponding individual cores} of two TNs with the same topology in order to obtain a \textit{new TN}, the cores of which carry both joint and individual information present in the original cores. Therefore, this is related to the recently introduced common feature extraction in  \cite{cobe}, however, unlike matrices, the proposed framework enables this to be performed on tensors of any order, with the only condition that the original tensors are represented as TNs with equivalent topologies. Practical advantages of the proposed framework are demonstrated through an image classification application based on the ETH-80 dataset \cite{eth}, whereby every dataset entry is represented as a TN, and their individual features are combined via our proposed framework. This allows for the extraction of the shared information in the original data, which is then fed to a Support Vector Machine algorithm (SVM), and attains an overall classification rate of $92.3\%$. The proposed framework therefore opens up new perspectives on the manipulation of TNs, completely removing the preconception that they have to be treated as stand-alone entities, and offers a new avenue for their applications.
%	
%	 It is shown that by mixing the features of the individual dataset entries via a sum of TNs, the shared information in the original data can be extracted in a way analogous to  the recently introduced common feature extraction in  \cite{cobe}.
	
%	 and this is shown to have an advantage over standard Support Vector Machine (SVM) and a related tensor-based classification algorithm.
	
%	This has a strong practical value when extracting features from a dataset, which are required to be common to all elements in the dataset itself \cite{cobe}.
%
%
%	This paper is organised as follows. In Section \ref{sec:theory} the notation and background are provided while the theory for sum of TNs is developed in Section \ref{sec:pepe}. Practical applications of the latter are given in Section \ref{sec:sim}, by classifying images from the ETH-80 dataset \cite{eth}. To further validate our framework, we compare results against the standard Support Vector Machine (SVM), and a related tensor-based classification algorithm. 

	\section{Notation and Background}\label{sec:theory}
	A tensor of order $N$ is denoted by boldface underlined uppercase letters, $\tensor{X} \in \mathbb{R}^{I_1 \times \dots \times I_N}$, a matrix by boldface uppercase letters, $\mathbf{X} \in \mathbb{R}^{J\times K}$, a vector by boldface lowercase letters, $\mathbf{x} \in \mathbb{R}^{N \times 1}$, and a scalar by italic lowercase letters,  $x \in \mathbb{R}$. Subscripts are generally described by indices $n,i,j,k$. An element in a $N$-th order tensor is denoted by $x_{i_1, i_2, \dots, i_N} = \tensor{X}(i_1, i_2, \dots, i_N)$. Given an $N$-th order tensor $\tensor{X} \in \mathbb{R}^{I_1 \times \dots \times I_n \times \dots \times I_N}$ and an $M$-th order tensor $\tensor{Y} \in \mathbb{R}^{J_1 \times \dots \times J_m \times \dots \times J_M}$, with $I_n = J_m$, their $(m,n)$-contraction product is $\tensor{Z} = \tensor{X} \times^m_n \tensor{Y}$, where $\tensor{Z}$ is an $(N+M-2)$-th order tensor (for more detail we refer to \cite{danilo_part1}).
	By convention $\times_n$ is equivalent to $\times^2_n$, and is referred to as mode-$n$ contraction. The mode-$n$ unfolding of a tensor $\tensor{X}$ rearranges its elements into a matrix, and is expressed as $\mathbf{X}_{(n)}$ (see \cite{dolgov_2014} for more details). The symbol $\otimes$ denotes the Kronecker product,  $\circ$ the outer product, and $||\cdot||$ the Frobenius norm. A TN representation of a tensor $\tensor{X}$ is denoted by a calligraphic bold letter, $\tn{X}$. Finally, the operator $\text{vec}(\cdot)$ indicates vectorization of a tensor.

	\subsection{Tucker Decomposition}\label{sec:tkd}

	The TKD is analogous to a higher order form of matrix factorization, and decomposes an original tensor $\tensor{X}$ into a core tensor contracted by a factor matrix along each corresponding mode  \cite{tucker1,Tuck1963a,tucker64extension}. In the case of a $3$-rd order tensor $\tensor{X}\in\mathbb{R}^{I\times J \times K}$, the TKD is expressed as
	\begin{equation}\label{eq:tkd}
	\begin{aligned}
	\tensor{X}& = \tensor{G} \times_1 \mathbf{A} \times_2 \mathbf{B} \times_3 \mathbf{C}\\
	& = \sum_{q}^{Q} \sum_{r}^{R} \sum_{p}^{P} g_{qrp} \mathbf{a}_{q} \circ \mathbf{b}_{r} \circ \mathbf{c}_{p}
	\end{aligned}
	\end{equation}
	where $\tensor{G} \in \mathbb{R}^{Q \times R \times P} $ and $\mathbf{A} \in \mathbb{R}^{I \times Q},\mathbf{B} \in \mathbb{R}^{J \times R}$, $\mathbf{C} \in \mathbb{R}^{K \times P}$, while a TKD for an $N$-th order tensor is given by

	%\begin{figure}[h]
	%	\centering
	%	\includegraphics[width=0.9\linewidth]{images/tucker}
	%	\caption{TKD applied on a $3$-rd order tensor $\tensor{X} \in \mathbb{R}^{I\times J \times K}$.}
	%	\label{fig:tkd3}
	%\end{figure}

	\begin{equation}\label{eq:tkdgen}
	\begin{aligned}
	\tensor{X} =& \tensor{G} \times_1 \mathbf{A}^{(1)} \times_2 \mathbf{A}^{(2)} \times_3 \dots \times_N \mathbf{A}^{(N)}
	\end{aligned}
	\end{equation}
	For convenience, any tensor expressed in the form of (\ref{eq:tkdgen}), will be referred to as ``in the TKD format", even though the factors $\mathbf{A}^{(n)}$ were not necessarily obtained via a TKD. The mode-$n$ unfolding of a tensor in the TKD format is
	\begin{equation}\label{eq:nthunfold}
	\mathbf{X}_{(n)}= \mathbf{A}^{(n)}\mathbf{G}_{(n)}(\mathbf{A}^{(N)}  \otimes \dots \otimes \mathbf{A}^{(n-1)} \otimes \mathbf{A}^{(n+1)} \otimes \dots \otimes \mathbf{A}^{(1)}   )^T
	\end{equation}

	\subsection{Background on Tensor Networks}

	A decomposition of a tensor into multi-way linked core tensors and matrices leads to  equations often involving numerous contraction products, which can be cumbersome to write and hard to visualize. For this reason it is common to represent tensors diagrammatically \cite{danilo_part1}, as in  Fig. \ref{fig:graph}.  An $N$-th order tensor can be represented as a node (circle) with as many edges (modes) as the tensor order. In TNs, contractions are designated by linking two common modes, called contraction modes, while ``dangling" edges are physical modes of the represented tensor.

	\begin{figure}[H]
		\centering
		\includegraphics[width=0.6\linewidth]{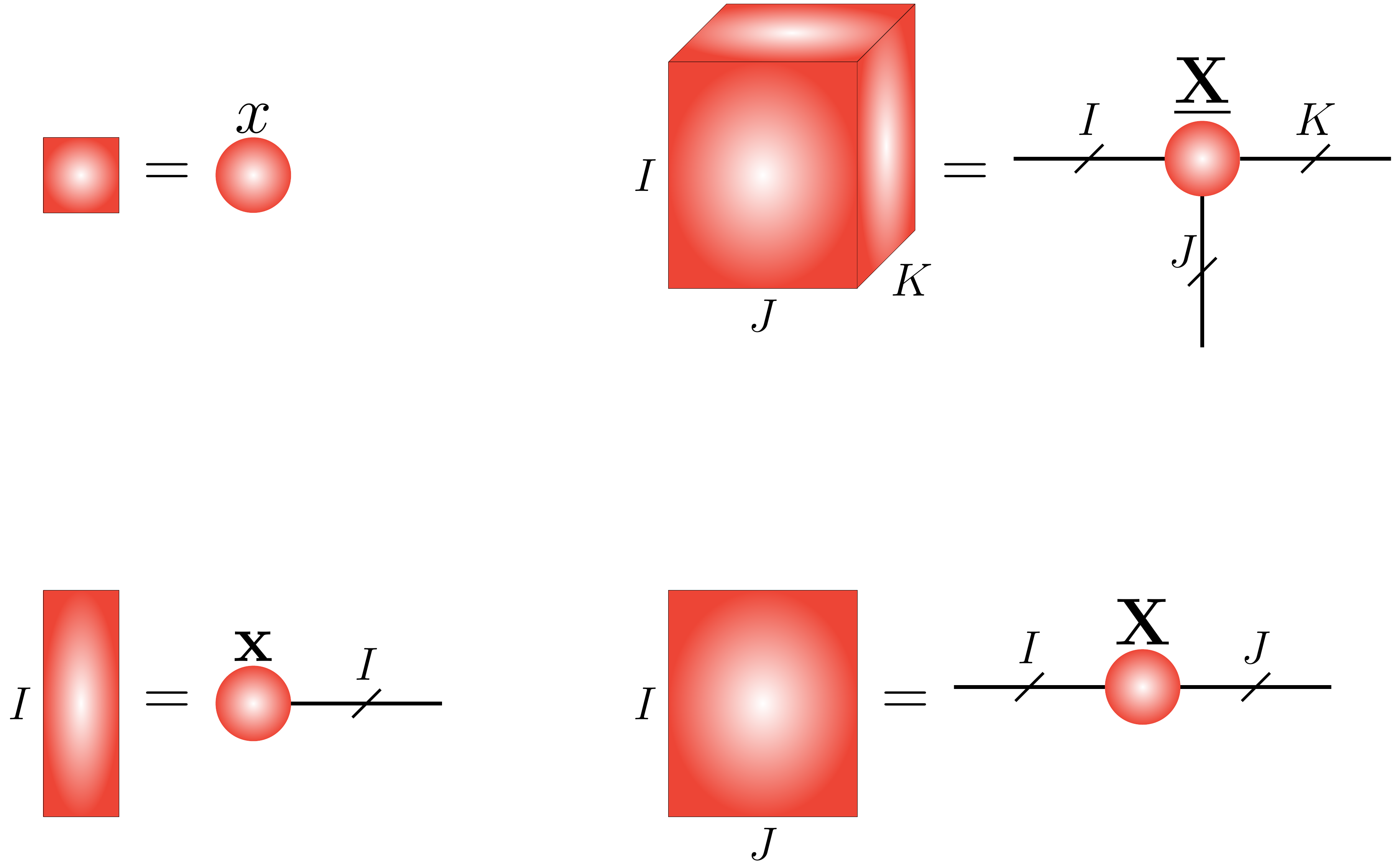}
		\caption{Building blocks of TNs. Anticlockwise from the top-left: scalar, vector, matrix,  and $3$-rd order tensor. Edges are called modes, and the associated label $I,J,K$, indicate their dimensionality.}
		\label{fig:graph}
	\end{figure}

	%\begin{figure}[h]
	%	\centering
	%	\includegraphics[width=0.9\linewidth]{images/TKD_rep}
	%	\caption{Graphical representation of a TKD on a $3$-rd order tensor $\tensor{X} \in \mathbb{R}^{I_1 \times I_2 \times I_3}$. Linking edges are the contracting modes, while those linked to a node only from one side represent the physical modes of the tensor $\tensor{X}$. \textcolor{red}{have the dimensions of the core changed to Q,R,P, same goes for the physical dimensions and the name of the matrices}}
	%	\label{fig:tkdgraph}
	%\end{figure}

	%\begin{figure}[h]
	%	\centering
	%	\includegraphics[width=0.9\linewidth]{images/TT_rep}
	%	\caption{TT representation of a 4th-order tensor $\tensor{Y} \in \mathbb{R}^{I_1 \times I_2 \times I_3 \times I_4}$. Notice that the only free modes are the physical ones.}
	%	\label{fig:tt}
	%\end{figure}

	Two examples of TNs are provided in Fig. \ref{fig:TN}, whereby any mode which is not a physical mode is a contraction mode. Nodes linked only to contraction modes are contraction nodes, whereas nodes linked to one or more physical mode are referred to as physical nodes. The ``shape" of a TN is its topology, where the concept of topology is the same to the one adopted in Graph Theory \cite{graphs}.

	Each node in a TN $\tn{X}$ can either represent a particular feature of the original tensor $\tensor{X}$ (in case of physical nodes), or a model on how features are combined (in case of contraction nodes). 
%	For example, in the general TKD in (\ref{eq:tkdgen}), matrix $\mathbf{A}^{(1)}$ describes the feature distribution along the first mode of $\tensor{X}$, matrix $\mathbf{A}^{(2)}$ along the second mode, and so on. 
	It is clear that, given that TNs represent tensors of any order,  without a framework for combining TNs, the problem of mixing features of individual cores, which subsequently allows for the extraction of the common features across the individual tensors, becomes difficult to address. This characteristic of ``feature locality" inherent to the nodes of a TN is of fundamental importance. Therefore our main motivation for this work is to provide the missing link for TN summation via a combination of their corresponding cores, thus simultaneously performing a mixture of features.

	\begin{figure}[H]
		\centering
		\includegraphics[width=0.5\linewidth]{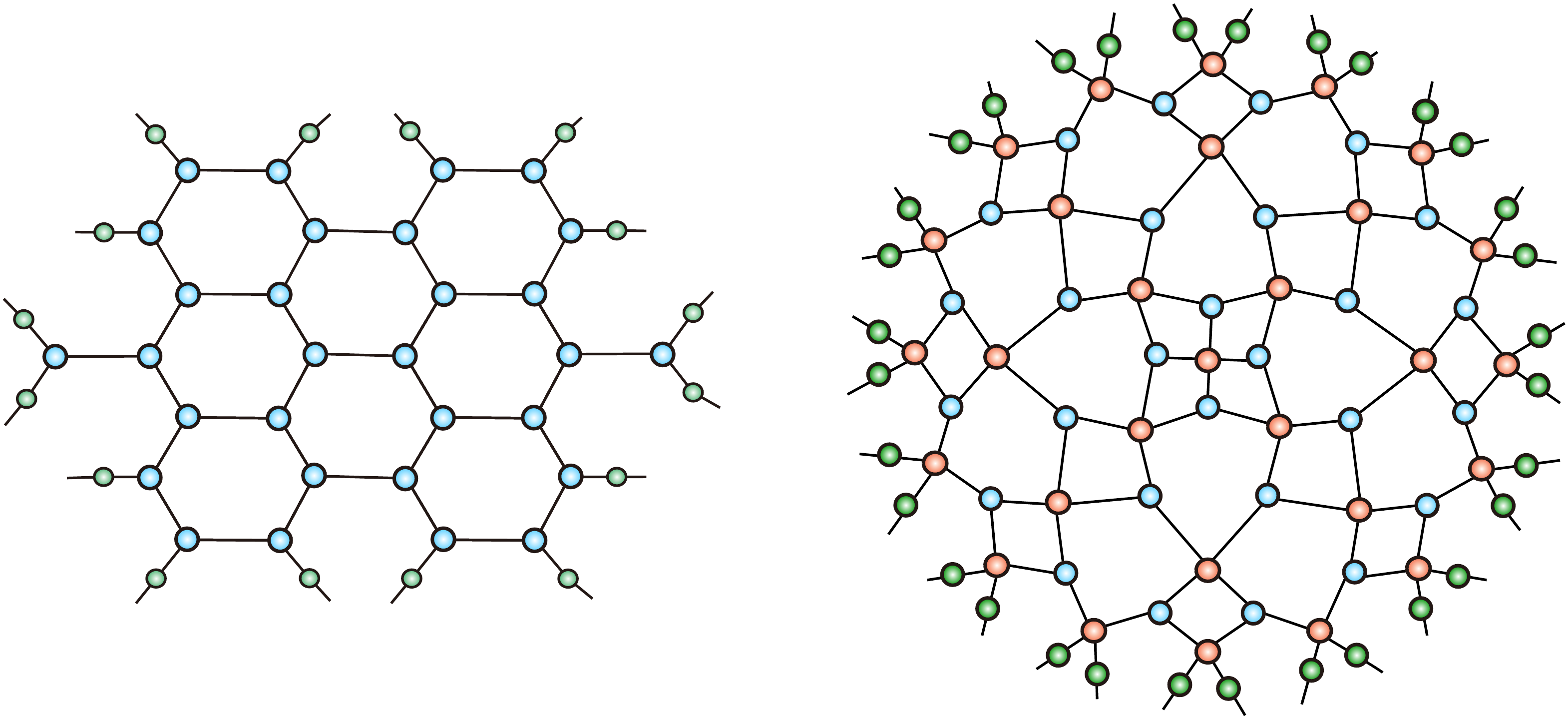}
		\caption{Examples of TNs (figure adapted from \cite{danilo_part1}). Left: a TN representing a $16$-th order tensor, with contraction nodes shown in blue, and physical nodes shown in green. Right: a TN representing a $32$-th order tensor, with contraction nodes shown in blue ($3$-rd order) and orange ($4$-th order), and physical nodes in green. }
		\label{fig:TN}
	\end{figure}

	\section{Sum of Tensor Networks}\label{sec:pepe}

	 To establish a framework of summation of two or more TNs, assume that the individual TNs: (i) have physical modes with equal dimensions, (ii)  have the same topology. We proceed by providing the definitions, the main conjecture, and specific propositions arising from the conjecture, followed by an outline of practical applications.

	\begin{definition}
		A \textbf{block tensor} is a tensor that is arranged into sub-tensors called \textbf{blocks}, that is, its entries are tensors of the same order but not necessarily of the same dimensionality.
	\end{definition}

	\begin{definition}
		The \textbf{superdiagonal} of a block tensor $\tensor{X}\in\mathbb{R}^{I_1 \times \dots \times I_N}$ is the collection of entries $x_{i_1,i_2,\dots, i_n}$, where $i_1 = i_2 = \dots = i_n$.
	\end{definition}

	\begin{conj}\label{conj:sumtns}
		Consider two tensors $\tensor{X}, \tensor{Y}\in \mathbb{R}^{I_1 \times \dots \times I_N}$ represented as TNs with equivalent topologies, but not necessarily with equivalent contracting modes, referred to as $\tn{X}$ and $\tn{Y}$. The sum $\tensor{Z} = \tensor{X}+\tensor{Y}$ can then be represented as a new TN, $\tn{Z}$, with equivalent topology to $\tn{X}$ and $\tn{Y}$. Its contraction nodes are in the form of a block tensor which is obtained by stacking the corresponding contraction nodes of $\tn{X}$ and $\tn{Y}$ along its superdiagonal. The physical nodes of $\tn{Z}$ are obtained by stacking the corresponding physical nodes of $\tn{X}$ and $\tn{Y}$ in such a way that the dimensionality of all contracting modes is increased but that of the physical modes is kept fixed.
	\end{conj}

	\begin{prop}\label{prop:chain}
		Conjecture \ref{conj:sumtns} holds for chains of matrices.
	\end{prop}

	\begin{proof}
		Consider Fig. \ref{fig:chain} as a graphical representation of Conjecture 1, and suppose $\mathbf{X}= \mathbf{A}_1 \mathbf{A}_2 \cdots \mathbf{A}_N$, $\mathbf{Y}= \mathbf{B}_1 \mathbf{B}_2 \cdots \mathbf{B}_N$, where $\mathbf{X}, \mathbf{Y} \in \mathbb{R}^{I \times J}$, and $\mathbf{A}_n, \mathbf{B}_n \in \mathbb{R}^{R_n \times R_{n+1}} $, with $R_0 = I, R_{N} = J$. Define a new chain of matrices $\mathbf{Z}=\mathbf{C}_1 \mathbf{C}_2 \cdots \mathbf{C}_N$, where each $\mathbf{C}_n$ is an arrangement of $\mathbf{A}_n, \mathbf{B}_n$ according to Conjecture \ref{conj:sumtns}. By a direct inspection of $\mathbf{Z}$, we obtain
		\begin{equation}
		\begin{aligned}
		\mathbf{Z}& =  \mathbf{C}_1 \mathbf{C}_2 \cdots \mathbf{C}_N \\
		&=\begin{bmatrix}
		\mathbf{A}_1 & \mathbf{B}_1
		\end{bmatrix}
		\begin{bmatrix}
		\mathbf{A}_2 & \mathbf{0} \\
		\mathbf{0}   & \mathbf{B}_2 \\
		\end{bmatrix} \cdots
		\begin{bmatrix}
		\mathbf{A}_{N-1} & \mathbf{0} \\
		\mathbf{0}   & \mathbf{B}_{N-1} \\
		\end{bmatrix}
		\begin{bmatrix}
		\mathbf{A}_N \\
		\mathbf{B}_N \\
		\end{bmatrix} \\
		&= \mathbf{A}_1 \mathbf{A}_2 \cdots \mathbf{A}_n +  \mathbf{B}_1 \mathbf{B}_2 \cdots \mathbf{B}_n \\
		&= \mathbf{X} + \mathbf{Y}
		\end{aligned}
		\end{equation}

	\end{proof}

	\begin{figure}[H]
		\centering
		\includegraphics[width=0.6\linewidth]{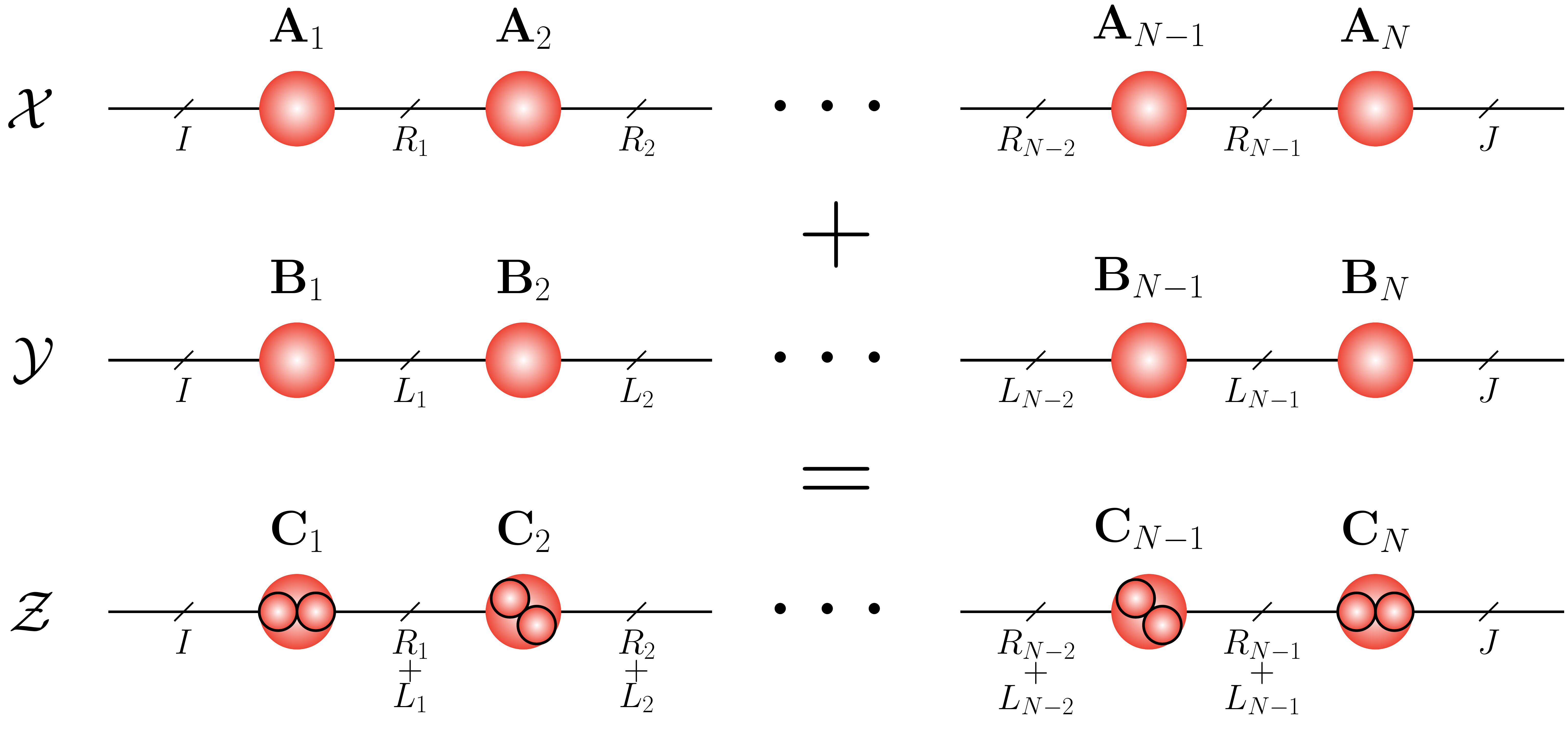}
		\caption{Graphical representation of a sum of matrices $\mathbf{X}$ and $\mathbf{Y}$,  expressed as a matrix chain. The matrices (nodes) $\mathbf{C}_n$ are composed of the matrices $\mathbf{A}_n$ and $\mathbf{B}_n$, which are arranged according to Conjecture \ref{conj:sumtns}, either through a concatenation (two horizontal subnodes) or a block-diagonal arrangement (two diagonal subnodes).}
		\label{fig:chain}
	\end{figure}

	\begin{prop}\label{prop:sumtkd}
		Conjecture \ref{conj:sumtns} holds for any tensor expressed in the TKD format.
	\end{prop}

	\begin{proof}
		For simplicity we here prove Proposition \ref{prop:sumtkd} for $3$-rd order tensors, but without loss of generality the result holds for any tensor order. Fig. \ref{fig:sumtkd} shows the tensors $\tensor{X}, \tensor{Y} \in \mathbb{R}^{I_1 \times I_2 \times I_3}$ in the TKD format, that is
		\begin{equation}
		\begin{aligned}
		\tensor{X} &= \tensor{G}_x \times_1 \mathbf{A}_x \times_2 \mathbf{B}_x \times_3 \mathbf{C}_x \\
		\tensor{Y} &= \tensor{G}_y \times_1 \mathbf{A}_y \times_2 \mathbf{B}_y \times_3 \mathbf{C}_y
		\end{aligned}
		\end{equation}
		with respective TN representations $\tn{X}, \tn{Y}$. A new TN, $\tn{Z}$, is obtained by combining $\tn{X}$ and $\tn{Y}$ according to Conjecture \ref{conj:sumtns} (observe Fig. \ref{fig:sumtkd}), and the so represented tensor, $\tensor{Z}$, can be described by
		\begin{equation}\label{eq:z}
		\tensor{Z} = \tensor{G}_z \times_1 \mathbf{A}_z \times_2 \mathbf{B}_z \times_3 \mathbf{C}_z
		\end{equation}
		The task is to show that $\tensor{Z}=\tensor{X}+\tensor{Y}$. We shall consider the mode-$1$ unfolding, but the same procedure can be applied to any mode. Define
%		\begin{equation}\label{eq:z2}
	$ %	\begin{aligned}
		\mathbf{A}_z =
		\begin{bmatrix}
		\mathbf{A}_x & \mathbf{A}_y
		\end{bmatrix},$
%		\mathbf{B}_z =
%		\begin{bmatrix}
%		\mathbf{B}_x & \mathbf{B}_y
%		\end{bmatrix}$
	%	\end{aligned}$
%		\end{equation}
		and $\tensor{G}_z$ as an arrangement of $\tensor{G}_x$ and $\tensor{G}_y$ along the superdiagonal of $\tensor{G}_z$, and consider
	
		\begin{equation}
		\begin{aligned}
		\mathbf{K}_1 &= (\mathbf{C}_x \otimes \mathbf{B}_x) ^T \\
		\mathbf{K}_2 &= (\mathbf{C}_y \otimes \mathbf{B}_y) ^T \\
		\end{aligned}
		\end{equation}
		 Upon performing the mode-$1$ unfolding of $\tensor{X}$ and $\tensor{Y}$ according to (\ref{eq:nthunfold}), and combining the matrices $\mathbf{X}_{(1)}, \mathbf{Y}_{(1)}$ in their TN form, then from Proposition \ref{prop:chain} the resulting matrix, $\mathbf{Z}_*$, can be described as
		\begin{equation}\label{eq:tnz}
		\mathbf{Z}_* =
		\begin{bmatrix}
		\mathbf{A}_x & \mathbf{A}_y
		\end{bmatrix}
		\begin{bmatrix}
		\mathbf{G}_{x(1)} & \mathbf{0} \\
		\mathbf{0} & \mathbf{G}_{y(1)}
		\end{bmatrix}
		\begin{bmatrix}
		\mathbf{K}_1 & \mathbf{K}_2
		\end{bmatrix}^T
		\end{equation}

		Therefore, in order to prove Proposition \ref{prop:sumtkd} it is sufficient to show that $\mathbf{Z}_{(1)}=\mathbf{Z}_*$, where $\mathbf{Z}_{(1)}$ is the mode-$1$ unfolding of $\tensor{Z}$ represented as in (\ref{eq:z}). To this end, consider
		\begin{equation} \label{eq:tkdz}
		\begin{aligned}
		&\mathbf{Z}_{(1)} = \mathbf{A}_z \mathbf{G}_{z(1)} (\mathbf{C}_z \otimes \mathbf{B}_z)^T\\
		&=
		\begin{bmatrix}
		\mathbf{A}_x & \mathbf{A}_y
		\end{bmatrix}
		\mathbf{G}_{z(1)} \big(
		\begin{bmatrix}
		\mathbf{C}_x & \mathbf{C}_y
		\end{bmatrix} \otimes
		\begin{bmatrix}
		\mathbf{B}_x & \mathbf{B}_y
		\end{bmatrix}
		\big)^T
		\end{aligned}
		\end{equation}

		For convenience, denote $\tensor{X}, \tensor{Y}, \tensor{G}_x, \tensor{G}_y \in \mathbb{R}^{2 \times 2 \times 2}$, and define $\tensor{G}_\alpha(:,:,j) = \mathbf{\hat{G}}_\alpha (j)$, where $\alpha\in\{x,y\}$. Hence
		\begin{equation}\label{eq:gz}
		\mathbf{G}_{z(1)}=
		\begingroup % keep the change local
		\setlength\arraycolsep{2pt}
		\begin{bmatrix}
		\mathbf{\hat{G}}_x (1) & \mathbf{0} & \mathbf{\hat{G}}_x (2) & \mathbf{0}& \mathbf{0}& \mathbf{0}& \mathbf{0}& \mathbf{0}\\
		\mathbf{0}& \mathbf{0}& \mathbf{0}& \mathbf{0}& \mathbf{0} & \mathbf{\hat{G}}_y (1) & \mathbf{0} &\mathbf{\hat{G}}_y (2)
		\end{bmatrix}
		\endgroup
		\end{equation}

		\noindent Without loss of generality assume
		\begin{equation}
		\begin{aligned}
		\mathbf{C}_x &=
		\begin{bmatrix}
		1 & 2 \\
		5 & 6
		\end{bmatrix} \hspace{5mm}
		\mathbf{C}_y =
		\begin{bmatrix}
		3 & 4 \\
		7 & 8
		\end{bmatrix}\\
		\end{aligned}
		\end{equation}
		to give
		\begin{equation}
		\begin{aligned}
		&\begin{bmatrix}
		\mathbf{C}_x & \mathbf{C}_y
		\end{bmatrix} \otimes
		\begin{bmatrix}
		\mathbf{B}_x & \mathbf{B}_y
		\end{bmatrix} = \\
		& =
		\begingroup % keep the change local
		\setlength\arraycolsep{2pt}
		\renewcommand*{\arraystretch}{1.5}
		\begin{bmatrix}
		1 \begin{bmatrix}
		\mathbf{B}_x & \mathbf{B}_y
		\end{bmatrix} & 2\begin{bmatrix}
		\mathbf{B}_x & \mathbf{B}_y
		\end{bmatrix} & 3\begin{bmatrix}
		\mathbf{B}_x & \mathbf{B}_y
		\end{bmatrix} & 4\begin{bmatrix}
		\mathbf{B}_x & \mathbf{B}_y
		\end{bmatrix} \\
		5\begin{bmatrix}
		\mathbf{B}_x & \mathbf{B}_y
		\end{bmatrix} & 6\begin{bmatrix}
		\mathbf{B}_x & \mathbf{B}_y
		\end{bmatrix} & 7\begin{bmatrix}
		\mathbf{B}_x & \mathbf{B}_y
		\end{bmatrix} & 8\begin{bmatrix}
		\mathbf{B}_x & \mathbf{B}_y
		\end{bmatrix}
		\end{bmatrix} 
		\endgroup\\
		 &= \mathbf{U}
		\end{aligned}
		\end{equation}

		Upon substituting $\mathbf{G}_{z(1)}$ and $\mathbf{U}$ into (\ref{eq:tkdz}) and making use of the sparse nature of (\ref{eq:gz}), we arrive at
		\begin{align}
		%\begin{aligned}
		&\mathbf{Z}_{(1)} = \mathbf{A}_z \mathbf{G}_{z(1)} \mathbf{U}^T \nonumber \\
				&=
				\begin{bmatrix}
				\mathbf{A}_x &\mathbf{A}_y
				\end{bmatrix}
				\begingroup % keep the change local
				\setlength\arraycolsep{2pt}
				\begin{bmatrix}
				\mathbf{\hat{G}}_x(1) &\mathbf{\hat{G}}_x(2)& \mathbf{0} &\mathbf{0} \nonumber  \\
				\mathbf{0} &\mathbf{0} & \mathbf{\hat{G}}_y(1) &\mathbf{\hat{G}}_y(2)
				\end{bmatrix}
				\endgroup
				%
				%
				%\hspace{55mm} 
				\begin{bmatrix}
				1 \mathbf{B}_x^T & 5\mathbf{B}_x^T \\
				2\mathbf{B}_x^T & 6\mathbf{B}_x^T \\
				3\mathbf{B}_y^T & 7\mathbf{B}_y^T\\
				4\mathbf{B}_y^T & 8\mathbf{B}_y^T
				\end{bmatrix}\\
		&=
		\begin{bmatrix}
		\mathbf{A}_x& \mathbf{A}_y
		\end{bmatrix}
		\begin{bmatrix}
		\mathbf{G}_{x(1)} & \mathbf{0} \\
		\mathbf{0} &\mathbf{G}_{y(1)}
		\end{bmatrix}
		\begin{bmatrix}
		\mathbf{K}_1^T\\
		\mathbf{K}_2^T
		\end{bmatrix} \nonumber \\
		&=
		\begin{bmatrix}
		\mathbf{A}_x &\mathbf{A}_y
		\end{bmatrix}
		\begin{bmatrix}
		\mathbf{G}_{x(1)} & \mathbf{0} \\
		\mathbf{0} &\mathbf{G}_{y(1)}
		\end{bmatrix}
		\begin{bmatrix}
		\mathbf{K}_1 & \mathbf{K}_2
		\end{bmatrix}^T = \mathbf{Z}_*
		%\end{aligned}\hspace{-20pt}
		\end{align}
	\end{proof}
	%	\begin{figure}[h]
	%		\centering
	%		\includegraphics[width=0.9\linewidth]{foo}
	%		\caption{Concatenation of $\tensor{G}_x$ and $\tensor{G}_y$ along the superdiagonal to form $\tensor{G}_z$}.
	%		\label{fig:super}
	%	\end{figure}

	\begin{figure}[H]
		\centering
		\includegraphics[width=0.6\linewidth]{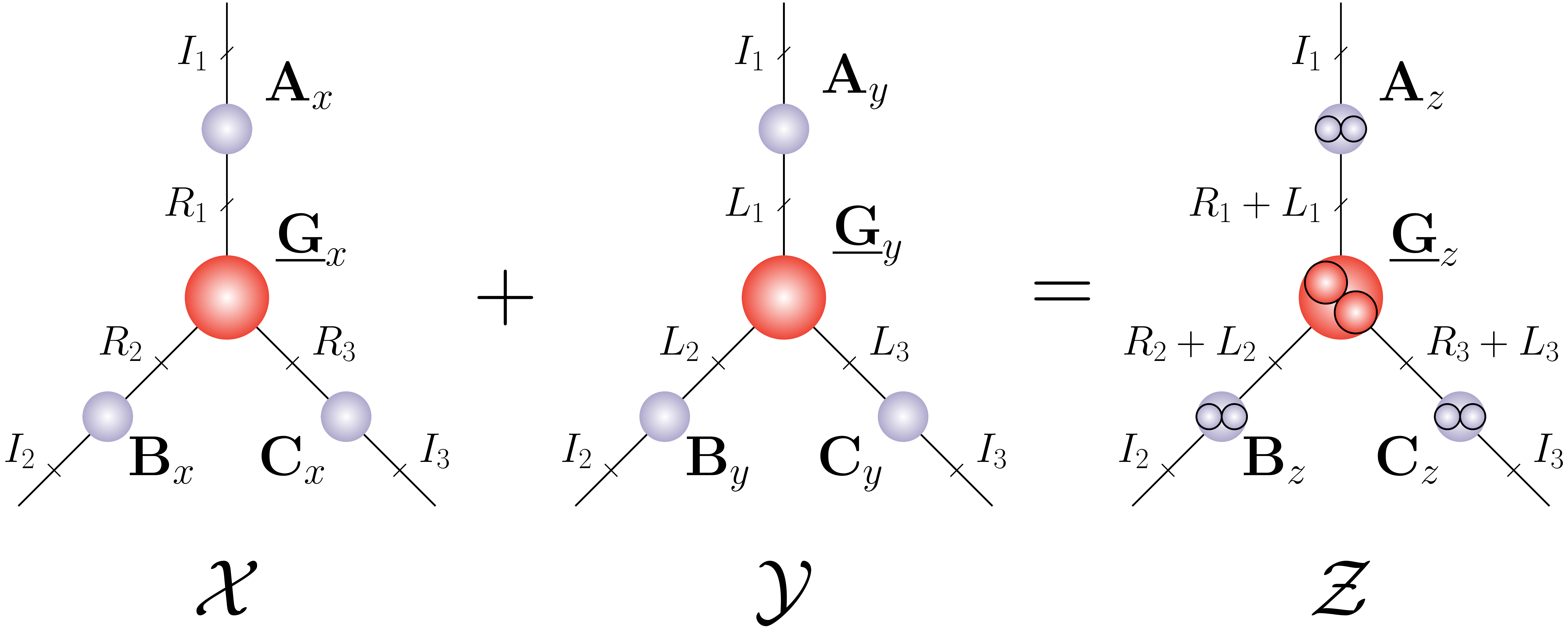}
		\caption{Graphical sum of tensors $\tensor{X}$ and $\tensor{Y}$ in the TKD format.}
		\label{fig:sumtkd}
	\end{figure}

	\section{Experimental Results}\label{sec:sim}

	Practical implications of the main contribution of this paper, that is, Proposition \ref{prop:sumtkd}, were investigated in an image classification task based on the benchmark ETH-80 dataset. It consists of 3280 images, composed of 8 classes, 10 objects per class, 41 images per object. For our simulations, the images were downsampled to $32 \times 32$ pixels. Given the RGB format of the considered images, our dataset consists of $3$-rd order tensors $\tensor{X}_m \in \mathbb{R}^{32 \times 32 \times 3}, m=1,\dots, M$ (where $M=3280$).

%	\begin{figure}[h]
%		\centering
%		\includegraphics[width=0.8\linewidth]{images/eth_400}
%		\caption{Sample of ETH-80 dataset. }
%		\label{fig:eth}
%	\end{figure}

	For each image, $\tensor{X}_m$, in the training set a TKD was performed by setting the size of the core tensors to $R_1 = R_2 = R_3 = 3$. The factor matrices $\{\mathbf{A}_m, \mathbf{B}_m, \mathbf{C}_m\}$ were scaled by $\eta^\frac{1}{3}$, where $\eta = ||\tensor{G}_m||$, to regularize the contribution of features. This is equivalent to normalizing the core tensor of each TKD decomposition to unit norm, without affecting the accuracy of the approximation. The scaled factor matrices $\{\mathbf{A}_m, \mathbf{B}_m, \mathbf{C}_m\}$ were concatenated into the matrices  $\{\mathbf{A}_d, \mathbf{B}_d, \mathbf{C}_d\}$. The SVD was subsequently applied and the first $\{R_1, R_2, R_3\}$ singular vectors were retained (we refer to this operator as tSVD($\cdot$)), which yielded matrices $\{\mathbf{A}_c, \mathbf{B}_c, \mathbf{C}_c\}$. Finally, for each image, $\tensor{X}_m$, a new core tensor was computed as
	\begin{equation}\label{eq:proj}
	\accentset{\circ}{\tensor{G}}_m = \tensor{X}_m \times_1 \mathbf{A}_c^T \times_2 \mathbf{B}_c^T \times_3 \mathbf{C}_c^T
	\end{equation}
	Equation (\ref{eq:proj}) represents the projection of the data images onto the features common to the full dataset, implying that $\accentset{\circ}{\tensor{G}}_m$ can be used for classification purposes. Their vectorized versions $\text{vec}(\accentset{\circ}{\tensor{G}}_m), m=1,\dots, M$,  were fed to a machine learning classifier in the form of an SVM (Gaussian kernel), which employed a one-vs-one (OVO) approach. During the testing stage, for each new element $\tensor{X}_*$, $\text{vec}(\accentset{\circ}{\tensor{G}}_*)$ was computed analogously and was classified according to the trained model, as summarized in Algorithm \ref{algo:pepe}.

	\begin{algorithm}
		\caption{Sum of TNs for image classification}
		\label{algo:pepe}
		\begin{algorithmic}[1]
			\State\textbf{Input:} Dataset $\{\tensor{X}_m\}_{m=1}^M$,  $\{R_1, R_2, R_3\}$
			\For {each element $m$ in dataset}
			\State $\tensor{X}_m = \tensor{G}_m \times_1 \times \mathbf{A}_m \times_2 \mathbf{B}_m \times_3 \mathbf{C}_m$
			\State $\eta = ||\tensor{G}_m||$
			\State $\mathbf{A}_d = 	\begin{bmatrix}
			\mathbf{A}_d &\eta^{\frac{1}{3}}  \mathbf{A}_m
			\end{bmatrix}$
			\State $\mathbf{B}_d = 	\begin{bmatrix}
			\mathbf{B}_d &\eta^{\frac{1}{3}} \mathbf{B}_m
			\end{bmatrix}$
			\State $\mathbf{C}_d = 	\begin{bmatrix}
			\mathbf{B}_d &\eta^{\frac{1}{3}} \mathbf{C}_m
			\end{bmatrix}$
			\EndFor
			\State $\mathbf{A}_c = \text{tSVD}(\mathbf{A}_d, R_1)$
			\State $\mathbf{B}_c = \text{tSVD}(\mathbf{B}_d, R_2)$
			\State $\mathbf{C}_c = \text{tSVD}(\mathbf{C}_d, R_3)$
			\For  {each element $m$ in dataset}
			\State$\accentset{\circ}{\tensor{G}}_m = \tensor{X}_m \times_1 \mathbf{A}_c^T \times_2 \mathbf{B}_c^T \times_3 \mathbf{C}_c^T$
			\EndFor
			\State Train classifier on $\{\text{vec}(\accentset{\circ}{\tensor{G}}_m)\}_{m=1}^M$.
		\end{algorithmic}
	\end{algorithm}

	The procedure outlined in Algorithm \ref{algo:pepe} was applied to the ETH-80 dataset, with $80\%$ of the available images serving as randomly selected training data. The average of 20 realizations yielded an accuracy of $92.3\%$. Observe in Fig. \ref{fig:confmat} that all classes were classified with a hit-rate $>90\%$ except for ``Cow", ``Dog", and ``Horse", which in the dataset indeed look similar (we refer to \cite{eth}).
	\begin{figure}[H]
		\centering
		\includegraphics[width=0.6\linewidth]{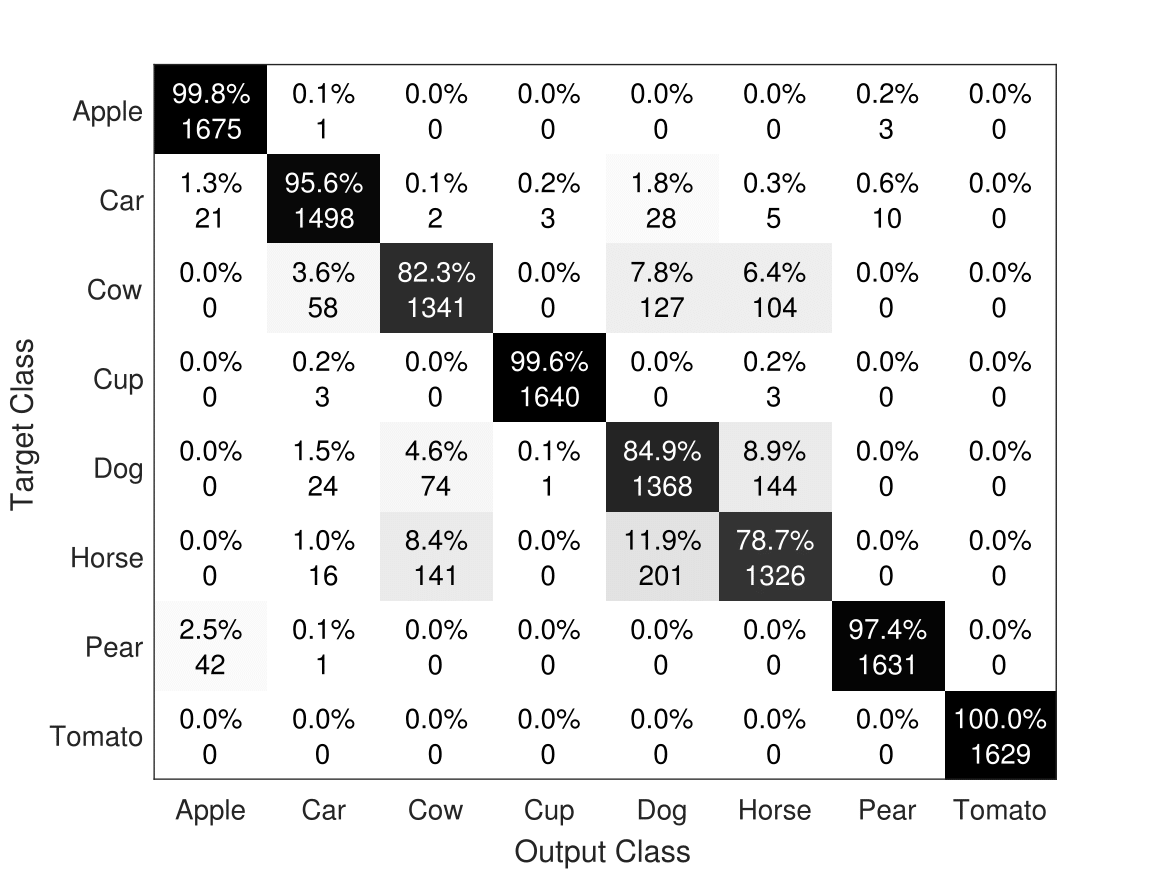}
		\caption{Confusion matrix of the classification algorithm, where $80\%$ of the data was used for training.}
		\label{fig:confmat}
	\end{figure}

	Performance comparisons were conducted against an SVM (Gaussian kernel) applied directly to the vectorised images, and a tensor-based classification method which we refer to as ``TKD-CONCAT", which concatenates all members of ETH-80 in a $4$-th order tensor, performs a TKD, and retains the first 3 factor matrices, as outlined in \cite{feature}. The results are shown in Fig. \ref{fig:diff_rates} and suggest that a direct summation of TNs yields a physically meaningful mixture of features, offering enhanced cross accuracy. Importantly, our proposed algorithm outperformed the other methods, especially when a small amount of data is available for training, as shown in the bottom graph. 
%	\begin{figure}[h]
%		\centering
%		\includegraphics[width=0.9\linewidth]{images/rate}
%		\caption{Accuracy rates of classification based on sum of TNs, vs. amount of data used for training. The proposed algorithm is compared against SVM, TKD-SUM, and TKD-CONCAT.}
%		\label{fig:rates}
%	\end{figure}
	\begin{figure}[H]
		\centering
		\includegraphics[width=0.6\linewidth]{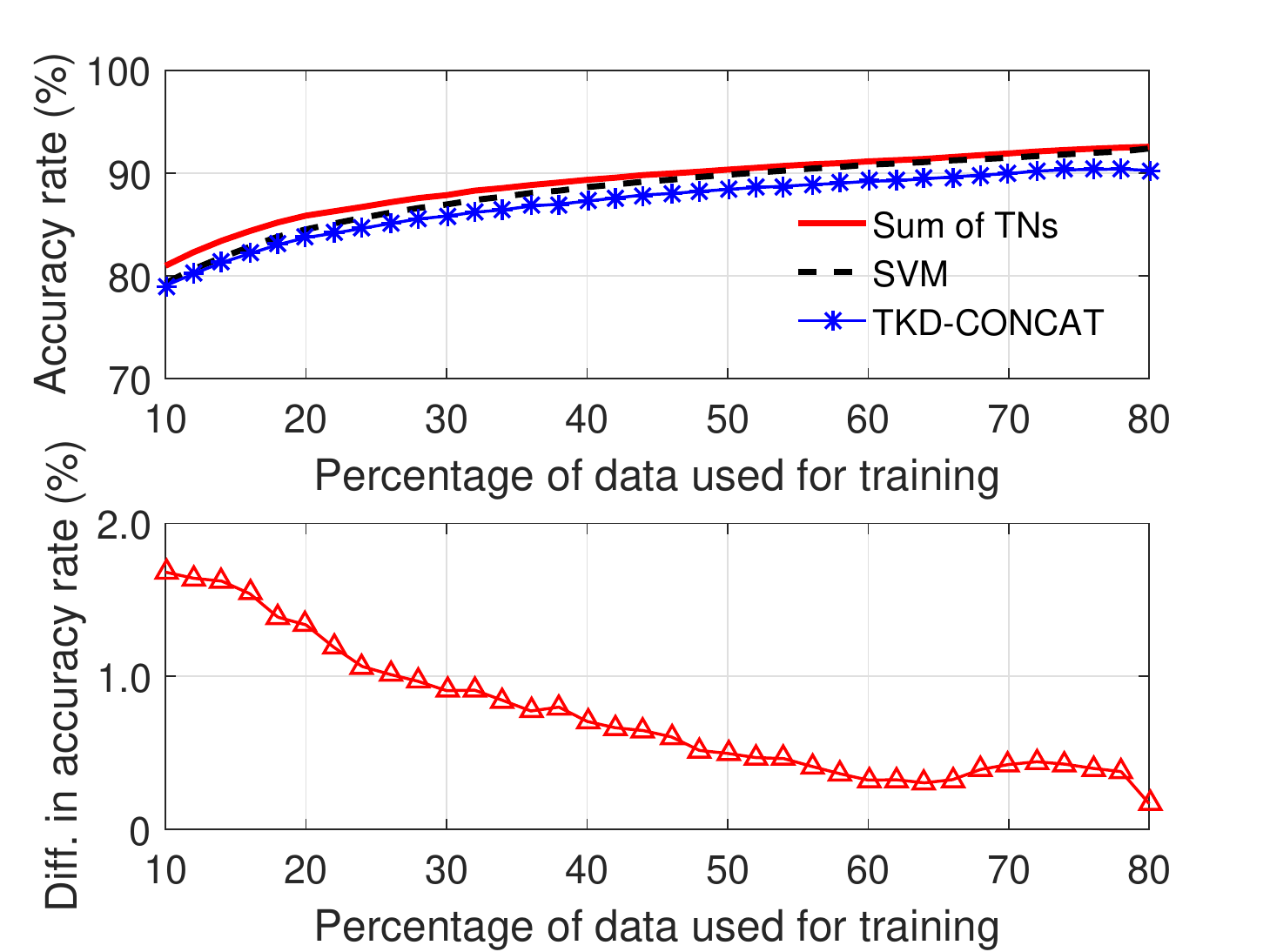}
		\caption{Classification results using the sum of TNs approach. Top: Accuracy rates. Bottom: Difference in accuracy between classification based on sum of TNs and standard SVM.}
		\label{fig:diff_rates}
	\end{figure}

	\section{Conclusion}\label{sec:conc}

	We have introduced a formalism behind the sum of tensor networks, and have validated the approach for chains of matrices ($2$-nd order tensors) and, more generally, for tensors in the Tucker format. By employing the analogy between the sum of tensor network cores and a feature fusion, we have devised a new  algorithm for image classification, which rests solely upon the the sum of tensor networks. Moreover, the proposed algorithm has been shown to exhibit a noticeable advantage when only few training data are available. Tests on the the ETH-80 dataset have attained an accuracy rate of $92.3\%$, and the proposed algorithm has been shown to outperform both standard Support Vector Machine and a related tensor-based classification approach. Generalisations of the proposed framework are the subject of ongoing work.

	%
	%\appendices
	%Some ext for the appendix.
	%
	%% use section* for acknowledgement
	%\section*{Acknowledgment}
	%





	%\bibliographystyle{IEEEtran}
	\bibliographystyle{IEEEtran}
	\bibliography{references}

\begin{thebibliography}{10}
\providecommand{\url}[1]{#1}
\csname url@samestyle\endcsname
\providecommand{\newblock}{\relax}
\providecommand{\bibinfo}[2]{#2}
\providecommand{\BIBentrySTDinterwordspacing}{\spaceskip=0pt\relax}
\providecommand{\BIBentryALTinterwordstretchfactor}{4}
\providecommand{\BIBentryALTinterwordspacing}{\spaceskip=\fontdimen2\font plus
\BIBentryALTinterwordstretchfactor\fontdimen3\font minus
  \fontdimen4\font\relax}
\providecommand{\BIBforeignlanguage}[2]{{%
\expandafter\ifx\csname l@#1\endcsname\relax
\typeout{** WARNING: IEEEtran.bst: No hyphenation pattern has been}%
\typeout{** loaded for the language `#1'. Using the pattern for}%
\typeout{** the default language instead.}%
\else
\language=\csname l@#1\endcsname
\fi
#2}}
\providecommand{\BIBdecl}{\relax}
\BIBdecl

\bibitem{comon}
J.~G. McWhirter and I.~Proudler, \emph{{Mathematics in signal processing
  V}}.\hskip 1em plus 0.5em minus 0.4em\relax Oxford University Press, 2002.

\bibitem{lathauwer_hooi}
L.~de~Lathauwer, B.~D. Moor, and J.~Vandewalle, ``{On the best rank-1 and
  rank-$(R_1,R_2, . . . ,R_N)$ approximation of higher-order tensors},''
  \emph{SIAM Journal on Matrix Analysis and Applications}, vol.~21, no.~4, pp.
  1324--1342, 2000.

\bibitem{lathauwer_hosvd}
------, ``{A multilinear singular value decomposition},'' \emph{SIAM Journal on
  Matrix Analysis and Applications}, vol.~21, no.~4, pp. 1253--1278, 2000.

\bibitem{cichocki_brain}
A.~Cichocki, ``{Tensors decompositions: new concepts for brain data
  analysis?}'' \emph{Journal of Control, Measurement, and System Integration
  (SICE)}, vol.~47, no.~7, pp. 507--517, 2011.

\bibitem{tenssignalprocessing}
A.~Cichocki, D.~P. Mandic, A.~H. Phan, C.~F. Caiafa, G.~Zhou, Q.~Zhao, and
  L.~D. Lathauwer, ``{Tensor decompositions for signal processing
  applications},'' \emph{IEEE Signal Processing Magazine}, vol.~32, no.~2, pp.
  145--163, 2015.

\bibitem{ttoseledets}
V.~Oseledets, ``{Tensor-train decomposition},'' \emph{Society for Industrial
  and Applied Mathematics}, vol.~33, no.~5, pp. 2295--2317, 2011.

\bibitem{parafac}
R.~Bro, ``{PARAFAC. Tutorial and applications},'' \emph{Chemometrics and
  Intelligent Laboratory Systems}, vol.~38, no.~2, pp. 149--171, 1997.

\bibitem{parafac2}
L.~de~Lathauwer, ``{A link between the canonical decomposition in multilinear
  algebra and simultaneous matrix diagonalization},'' \emph{SIAM Journal on
  Matrix Analysis and Applications}, vol.~28, no.~7, pp. 642--666, 2006.

\bibitem{tucker1}
L.~R. Tucker, ``{Some mathematical notes on three-mode factor analysis},''
  \emph{Psychometrika}, vol.~31, no.~3, pp. 279--311, 1966.

\bibitem{tucker2}
------, ``{Tucker dimensionality reduction of three-dimensional arrays in
  linear time},'' \emph{SIAM Journal on Matrix Analysis and Applications},
  vol.~30, no.~3, pp. 939--956, 2008.

\bibitem{cobe}
G.~Zhou, A.~Cichocki, Y.~Zhang, and D.~Mandic, ``{Group component analysis for
  multiblock data: common and individual feature extraction},'' \emph{IEEE
  Transactions on Neural Networks and Learning Systems}, vol.~27, no.~11, pp.
  2426--2439, 2016.

\bibitem{eth}
B.~Leibe and B.~Schiele, ``Analyzing appearance and contour based methods for
  object categorization,'' in \emph{Proc. IEEE Conference on Computer Vision
  and Pattern Recognition (CVPR’03)}, June 2003, pp. 409--415.

\bibitem{danilo_part1}
A.~Cichocki, I.~Oseledets, Q.~Zhao, N.~Lee, A.~H. Phan, and D.~Mandic,
  ``{Tensor networks for dimensionality reduction and large-scale pptimization.
  Part: 1 low-rank tensor decompositions},'' vol.~9, no. 4--5, pp. 249--429,
  2016.

\bibitem{dolgov_2014}
S.~Dolgov and D.~Savostyanov, ``{Alternating minimal energy methods for linear
  systems in higher dimensions},'' \emph{SIAM Journal on Scientific Computing},
  vol.~36, no.~5, pp. A2248--A2271, 2014.

\bibitem{Tuck1963a}
L.~R. Tucker, ``Implications of factor analysis of three-way matrices for
  measurement of change,'' in \emph{{P}roblems in Measuring Change}, C.~W.
  Harris, Ed.\hskip 1em plus 0.5em minus 0.4em\relax Madison WI: University of
  Wisconsin Press, 1963, pp. 122--137.

\bibitem{tucker64extension}
------, ``{T}he extension of factor analysis to three-dimensional matrices,''
  in \emph{{C}ontributions to Mathematical Psychology.}, H.~Gulliksen and
  N.~Frederiksen, Eds.\hskip 1em plus 0.5em minus 0.4em\relax New York: Holt,
  Rinehart and Winston, 1964, pp. 110--127.

\bibitem{graphs}
R.~Trudeau, \emph{{Introduction to graph theory}}.\hskip 1em plus 0.5em minus
  0.4em\relax Kent State University Press, 1993.

\bibitem{feature}
A.~H. Phan and A.~Cichocki, ``{Tensor decompositions for feature extraction and
  classification of high dimensional datasets},'' \emph{IEICE Nonlinear Theory
  and Its Applications}, vol.~1, no.~1, pp. 37--68, 2010.

\end{thebibliography}

	%\begin{biography}[{\includegraphics[width=1in,height=1.25in,clip,keepaspectratio]{mshell}}]{Michael Shell}


	%\vfill



\end{document}